\documentclass[11pt,letterpaper]{article}

\usepackage[utf8]{inputenc}
\usepackage[T1]{fontenc}
\usepackage{lmodern}
\usepackage[DIV=11]{typearea} 

\usepackage{microtype}

\usepackage{amssymb}
\usepackage{amsmath}
\usepackage{amsthm}
\usepackage{thmtools}
\usepackage{thm-restate}

\usepackage[procnumbered,ruled,vlined,linesnumbered]{algorithm2e}

\usepackage{xcolor}
\usepackage{xspace}
\usepackage{xfrac}

\usepackage{comment}

\usepackage[backend=biber, style=alphabetic, backref=true, doi=false, url=false, maxcitenames=3, mincitenames=3, maxbibnames=10, minbibnames=10, sortlocale=en_US]{biblatex}  

\usepackage[ocgcolorlinks]{hyperref} 
\usepackage{cleveref}

\allowdisplaybreaks[1]

\makeatletter
\g@addto@macro\bfseries{\boldmath}
\makeatother

\makeatletter
\g@addto@macro\mdseries{\unboldmath}
\g@addto@macro\normalfont{\unboldmath}
\g@addto@macro\rmfamily{\unboldmath}
\g@addto@macro\upshape{\unboldmath}
\makeatother

\DeclareCiteCommand{\citem}
    {}
    {\mkbibbrackets{\bibhyperref{\usebibmacro{postnote}}}}
    {\multicitedelim}
    {}

\renewcommand*{\multicitedelim}{\addcomma\space}

\AtEveryCitekey{\clearfield{note}}

\newcommand{\myhref}[1]{%
  \iffieldundef{doi}
    {\iffieldundef{url}
       {#1}
       {\href{\strfield{url}}{#1}}}
    {\href{http://dx.doi.org/\strfield{doi}}{#1}}%
}

\DeclareFieldFormat{title}{\myhref{\mkbibemph{#1}}}
\DeclareFieldFormat
  [article,inbook,incollection,inproceedings,patent,thesis,unpublished]
  {title}{\myhref{\mkbibquote{#1\isdot}}}

\makeatletter
\AtBeginDocument{%
    \newlength{\temp@x}%
    \newlength{\temp@y}%
    \newlength{\temp@w}%
    \newlength{\temp@h}%
    \def\my@coords#1#2#3#4{%
      \setlength{\temp@x}{#1}%
      \setlength{\temp@y}{#2}%
      \setlength{\temp@w}{#3}%
      \setlength{\temp@h}{#4}%
      \adjustlengths{}%
      \my@pdfliteral{\strip@pt\temp@x\space\strip@pt\temp@y\space\strip@pt\temp@w\space\strip@pt\temp@h\space re}}%
    \ifpdf
      \typeout{In PDF mode}%
      \def\my@pdfliteral#1{\pdfliteral page{#1}}
      \def\adjustlengths{}%
    \fi
    \ifxetex
      \def\my@pdfliteral #1{}
      \def\adjustlengths{\setlength{\temp@h}{-\temp@h}\addtolength{\temp@y}{1in}\addtolength{\temp@x}{-1in}}%
    \fi%
    \def\Hy@colorlink#1{%
      \begingroup
        \ifHy@ocgcolorlinks
          \def\Hy@ocgcolor{#1}%
          \my@pdfliteral{q}%
          \my@pdfliteral{7 Tr}
        \else
          \HyColor@UseColor#1%
        \fi
    }%
    \def\Hy@endcolorlink{%
      \ifHy@ocgcolorlinks%
        \my@pdfliteral{/OC/OCPrint BDC}%
        \my@coords{0pt}{0pt}{\pdfpagewidth}{\pdfpageheight}%
        \my@pdfliteral{F}
        \my@pdfliteral{EMC/OC/OCView BDC}%
        \begingroup%
          \expandafter\HyColor@UseColor\Hy@ocgcolor%
          \my@coords{0pt}{0pt}{\pdfpagewidth}{\pdfpageheight}%
          \my@pdfliteral{F}
        \endgroup%
        \my@pdfliteral{EMC}%
        \my@pdfliteral{0 Tr}
        \my@pdfliteral{Q}%
      \fi
      \endgroup
    }%
}
\makeatother

\colorlet{DarkRed}{red!50!black}
\colorlet{DarkGreen}{green!50!black}
\colorlet{DarkBlue}{blue!50!black}

\hypersetup{
	linkcolor = DarkRed,
	citecolor = DarkGreen,
	urlcolor = DarkBlue,
	bookmarksnumbered = true,
	linktocpage = true
}

\declaretheorem[numberwithin=section]{theorem}
\declaretheorem[numberlike=theorem]{lemma}

\declaretheorem[numberlike=theorem]{corollary}

\declaretheorem[numberlike=theorem]{invariant}

\newcommand{\nil}{\bot}

\newcommand{\Prepro}{\ensuremath{\textsc{DetPreprocessing}(G, \tau, h)}\xspace}

\newcommand{\ExtendDists}{\ensuremath{\textsc{DetExtDistances}(\Pi, h)}\xspace}

\title{Fully-Dynamic All-Pairs Shortest Paths: Improved Worst-Case Time and Space Bounds}
\author{
Maximilian Probst Gutenberg\thanks{BARC, University of Copenhagen, Universitetsparken 5, Copenhagen 2100, Denmark, The author is supported by Basic Algorithms Research Copenhagen (BARC), supported by Thorup's Investigator Grant from the Villum Foundation under Grant No. 16582.}
\and
Christian Wulff-Nilsen\thanks{Department of Computer Science, University of Copenhagen. This research is supported by the Starting Grant 7027-00050B from the Independent Research Fund Denmark under the Sapere Aude research career programme.
}
}

\date{\today}

\hypersetup{
	pdftitle = {Fully-Dynamic All-Pairs Shortest Paths: Improved Worst-Case Time and Space Bounds},
	pdfauthor = {Maximilian Probst Gutenberg, Christian Wulff-Nilsen}
}

\begin{document}
\maketitle
\begin{abstract}
Given a directed weighted graph $G=(V,E)$ undergoing vertex insertions \emph{and} deletions, the All-Pairs Shortest Paths (APSP) problem asks to maintain a data structure that processes updates efficiently and returns after each update the distance matrix to the current version of $G$. In two breakthrough results, Italiano and Demetrescu [STOC '03] presented an algorithm that requires $\tilde{O}(n^2)$ \emph{amortized} update time, and Thorup showed in [STOC '05] that \emph{worst-case} update time $\tilde{O}(n^{2+3/4})$ can be achieved. In this article, we make substantial progress on the problem. We present the following new results:
\begin{itemize}
    \item We present the first deterministic data structure that breaks the $\tilde{O}(n^{2+3/4})$ worst-case update time bound by Thorup which has been standing for almost 15 years. We improve the worst-case update time to $\tilde{O}(n^{2+5/7}) = \tilde{O}(n^{2.71..})$ and to $\tilde{O}(n^{2+3/5}) = \tilde{O}(n^{2.6})$ for unweighted graphs.
    \item We present a simple deterministic algorithm with $\tilde{O}(n^{2+3/4})$ worst-case update time ($\tilde{O}(n^{2+2/3})$ for unweighted graphs), and a simple Las-Vegas algorithm with worst-case update time $\tilde{O}(n^{2+2/3})$ ($\tilde{O}(n^{2 + 1/2})$ for unweighted graphs) that works against a non-oblivious adversary. Both data structures require space $\tilde{O}(n^2)$. These are the first exact dynamic algorithms with truly-subcubic update time \emph{and} space usage. This makes significant progress on an open question posed in multiple articles [COCOON'01, STOC '03, ICALP'04, Encyclopedia of Algorithms '08] and is critical to algorithms in practice [TALG '06] where large space usage is prohibitive. Moreover, they match the worst-case update time of the best previous algorithms and the second algorithm improves upon a Monte-Carlo algorithm in a weaker adversary model with the same running time [SODA '17].
\end{itemize}
\end{abstract}

\pagebreak

\section{Introduction}\label{sec:Intro}

The All-Pairs Shortest Paths problem is one of the most fundamental algorithmic problems and is commonly taught in undergraduate courses to every Computer Science student. Whilst static algorithms for the problem are well-known for several decades, the dynamic versions of the problem have recently received intense attention by the research community. In the dynamic setting, the underlying graph $G=(V,E)$ undergoes updates, most commonly edge insertions and/or deletions. Most dynamic All-Pairs Shortest Paths algorithms can further handle vertex insertions (with up to $n$ incident edges) and/or deletions. 

\paragraph{The problem.} In this article, we are only concerned with the fully-dynamic All-Pairs Shortest Path (APSP) problem with worst-case update time, i.e. given a fully-dynamic graph $G=(V,E)$, undergoing vertex insertions and deletions, we want to guarantee minimal update time after each vertex update to recompute the distance matrix of the new graph. This is opposed to the version of the problem that allows for amortized update time. Moreover, we focus on space-efficient data structures and show that our data structures even improve over the most space-efficient APSP algorithms with amortized update time. We further point out, that for the fully-dynamic setting, vertex updates are more general than edge updates, since any edge update can be simulated by a constant number of vertex updates.

\paragraph{Related Work.} The earliest partially-dynamic algorithm to the All-Pairs Shortest Path problem is most likely the algorithm by Johnson \cite{johnson1977efficient} that can be easily extended to handle vertex insertions in $O(n^2)$ worst-case update time per insertion given the distance matrix of the current graph. The first fully-dynamic algorithm was presented by King \cite{king1999fully} with $O(n^{2.5}\sqrt{W \log n})$ amortized update time per edge insertion/deletion where $W$ is the largest edge weight building upon a classic data structure for decremental Single-Source Shortest Paths by Even and Shiloach \cite{even1979line}. Later, King and Thorup \cite{king2001space} improved the space bound to $\tilde{O}(n^{2.5}\sqrt{W})$. In follow-up work by Demetrescu and Italiano \cite{demetrescu2002improved, demetrescu2006fully}, the result was generalized to real edge weights with the same bounds. In 2004, Demetrescu and Italiano\cite{demetrescu2004new} presented a new approach to the All-Pairs Shortest Paths problems that only requires $\tilde{O}(n^2)$ amortized update time for vertex insertions/deletions using $O(n^3)$ space. Thorup improved and simplified the approach in \cite{thorup2004fully} and even extended it to handle negative edge weights. Based on this data structure, he further developed the first data structure with $\tilde{O}(n^{2+3/4})$ worst-case update time \cite{thorup2005worst} for vertex insertion/deletions improving over the trivial $O(n^3)$ worst-case update time that can be obtained by recomputation. However, his data structure requires \emph{supercubic} space in $n$. 

Abraham, Chechik and Krinninger \cite{AbrahamCK17} showed that using randomization a Monte Carlo data structure can be devised with worst-case update time $\tilde{O}(n^{2+2/3})$. For unweighted graphs, they further obtain worst-case update time $\tilde{O}(n^{2+1/2})$. Both algorithms require $O(n^3)$ space since they require a list of size $O(n)$ for each pair of vertices. Both algorithms work against an oblivious adaptive adversary, that is the adversary can base the update sequence on the output produced by the algorithm but has no access to the random choices that the algorithm makes. A drawback of the algorithm is that if it outputs an incorrect shortest path (which it does with probability at most $1-n^{-c}$ for some constant $c > 0$), the adversary can exploit the revealed information and compromise the data structure for up to $n^{1/3}$ updates before their data structure is recomputed.

We also point out that the problem of Approximate All-Pairs Shortest Paths was solved in various dynamic graph settings \cite{baswana2002improved, demetrescu2004new, roditty2004dynamic, bernstein2009fully, roditty2012dynamic, abraham2013dynamic, bernstein2016maintaining, henzinger2016dynamic, shiriFocs, brand2019dynamic, probstWulffNilsenDetSSSP}. In the setting of $(1+\epsilon)$-approximate shortest paths, the best algorithms achieve amortized update time $\tilde{O}(m/\epsilon)$. However, the only of these algorithms that gives a better guarantee than the trivial $\tilde{O}(mn/\epsilon)$ on the worst-case update time is the algorithm in \cite{brand2019dynamic} that achieves time $\tilde{O}(n^{2.045}/\epsilon^2)$ for directed graphs with positive edge weights relies on fast matrix multiplication.

\paragraph{Our results.} 

We present the first deterministic data structure that breaks Thorup's longstanding bound of $\tilde{O}(n^{2+3/4})$ worst-case update time.

\begin{theorem}
Let $G$ be an $n$-vertex directed edge-weighted graph undergoing vertex insertions and deletions. Then there exists a deterministic data structure which can maintain distances in $G$ between all pairs of vertices in worst-case update time $O(n^{19/7}(\log n)^{8/7})$. If the graph is unweighted, the running time can be improved to $O(n^{2.6}\log n)$.
\end{theorem}

Further, we present the first algorithm for the fully-dynamic All-Pairs Shortest Paths problem (even amortized) in weighted graphs that obtains truly sub-cubic time and space usage at the same time\footnote{For small weights, i.e. weights in $n^{1-\epsilon}$ for some $\epsilon > 0$, the algorithm \cite{demetrescu2006fully} gives subcubic update time and space, both bound by $\tilde{O}(n^{2.5}\sqrt{W})$. However, as pointed out in \cite{demetrescu2006experimental}, real-world graphs often have large edge weights (for example, the internet graph had largest weight of roughly $10^4$ in 2006)}. Further, this is also the first algorithm that breaks the space/update-time product of $\Omega(n^5)$ which stood up to this article even for unweighted, undirected graphs. We hope that this gives new motivation to study amortized fully-dynamic algorithms that achieve $\tilde{O}(n^2)$ update-time \emph{and} space which is a central open question in the area, posed in \cite{demetrescu2004new, thorup2004fully, demetrescu2006experimental, Italiano2008} and has practical importance. 

\begin{theorem}
Let $G$ be an $n$-vertex directed edge-weighted graph undergoing vertex insertions and deletions. Then there exists a deterministic data structure which can maintain distances in $G$ between all pairs of vertices in worst-case update time $O(n^{2+3/4} (\log n)^{2/3})$ using space $\tilde{O}(n^2)$. If the graph is unweighted, the running time can be improved to $O(n^{2+2/3} (\log n)^{2/3})$.
\end{theorem}

Finally, we present a data structure that is randomized and matches the update times achieved in \cite{AbrahamCK17} up to polylogarithmic factors. However, their data structure is Monte-Carlo, and our data structure uses only $\tilde{O}(n^2)$ compared to $\tilde{O}(n^3)$ space and is slightly more robust, i.e. the data structure in \cite{AbrahamCK17} works against an adaptive adversary and therefore the adversary can base its updates on the output of the algorithm, whilst our algorithm works against a non-oblivious adversary that is the adversary also has access to the random bits used throughout the execution of the algorithm\footnote{The former model assumes for example that the adversary cannot use information about the running time of the algorithm during each update whilst we do not require this assumption.}.

\begin{theorem}
Let $G$ be an $n$-vertex directed edge-weighted graph undergoing vertex insertions and deletions. Then, there exists a Las-Vegas data structure which can maintain distances in $G$ between all pairs of vertices with update time $O(n^{2+2/3} (\log n)^3)$ w.h.p. using space $\tilde{O}(n^2)$ against a non-oblivious adversary. If the graph is unweighted, the running time can be improved to $O(n^{2+1/2} (\log n)^{3})$.
\end{theorem}

\paragraph{Our Techniques.} We focus on the decremental problem that we then generalize to the fully-dynamic setting using Johnson's algorithm. The most crucial ingredient of our new decremental algorithms is a new way to use congestion: for each shortest path $\pi_{s,t}$ from $s$ to $t$, each vertex on the shortest path is assigned a congestion value that relates to the costs induced by a deletion of such a vertex. If a vertex participates in many shortest paths, its deletion is expensive since we need to restore all shortest paths in which it participated. Thus, if the congestion of a vertex $v$ accumulated during some shortest path computations is too large, we simply remove the vertex from the graph and continue our shortest path computations on the graph $G \setminus \{v\}$. We defer handling the vertices of high congestion to a later stage and prepare for their deletion more carefully. This differs significantly from previous approaches that compute all paths in a specific order to avoid high congestion. Our new approach is simpler, more flexible and can be used to avoid vertices even at lower thresholds. 

The second technique we introduce is to use separators to recompute shortest paths after a vertex deletion. This allows us to speed up the computation since we can check fewer potential shortest paths. Since long paths have better separators, we can reduce the congestion induced by these paths and therefore reduce the overall amount of congestion on all vertices.

Once we complete our shortest path computations, we obtain the set of highly congested vertices and handle them using a different approach presented by Abraham, Chechik and Krinninger \cite{AbrahamCK17} that maintains deterministically the shortest paths through these vertices. These are exactly the shortest paths that we might have missed in the former step when we removed congested vertices. Thus, taking the paths of smaller weight, we obtain the real shortest paths in $G$.

Finally, we present a randomized technique to layer our approach where we use a low congestion threshold initially to identify uncritical vertices and then obtain with each level in the hierarchy a smaller set of increasingly critical vertices that require more shortest path computations on deletion. Since the sets of critical vertices are decreasing in size, we can afford to invest more update time in the maintenance of shortest paths through these vertices.

\section{Preliminaries}\label{sec:Prelim}

We denote by $G=(V,E, w)$ the input digraph where $w$ is the weight function mapping each edge to a number in the reals and define $n=|V|$ and $m = |E|$. In this article, we define $H \subseteq G$ to refer to $H$ being a vertex-induced subgraph of $G$, i.e. $H = G[V \setminus D]$ for some $D \subseteq V$. We also slightly abuse notation and write $G \setminus D$ for $D \subseteq V$ to denote $G[V \setminus D]$. We let the graph $H$ with edge directions reversed be denote by $\overleftarrow{H}$.

The \emph{weight} of a path $\pi$ in an edge-weighted graph $G$ is the sum of weights of its edges. We let $w$ denote the weight function that maps each path in $\pi$ to its weight. We use $\nil$ to denote the empty path of weight $\infty$. Given two paths $\pi_1 = \langle u_1, u_2, \ldots, u_p \rangle$ and $\pi_2 = \langle v_1, v_2 \ldots, v_q \rangle$ in $G$ where $u_p = v_1$, denote by $\pi_1 \circ \pi_2$ the concatenated path $\langle u_1, u_2 \ldots, u_p =v_1, v_2, \ldots, v_q\rangle$. For any path $\pi$, we define $\pi\circ\nil = \nil\circ \pi = \nil$.

Let $\pi_{s,t}$ be a path starting in vertex $s$ and ending in vertex $t$. Then $\pi_{s,t}$ is a \emph{shortest path} in $G$ if its sum of edge weights is minimized over all paths from $s$ to $t$ in $G$. We denote the weight of a shortest path from $s$ to $t$ by $\mathbf{dist}_G(s,t)$. 

We say $\pi_{s,t}$ is the shortest path from $s$ to $t$ \emph{through} $C \subset V$, if $\pi_{s,t}$ is the path of minimum weight from $s$ to $t$ that contains a vertex in $C$. We further say a path $\pi_{s,t}$ has hop $h$ or is a \emph{$h$-hop-restricted path} in $G$ if it consists of at most $h$ edges. We denote by $\mathbf{dist}_G^h(s,t)$ the weight of the $h$-hop-restricted shortest path from $s$ to $t$. Finally, we define the notion of an \emph{improving} shortest path $\pi_{s,t}$ in $G$ with regard to $H \subseteq G$ to be a path of weight at most $\mathbf{dist}_{H}(s,t)$. We often combine these notions, saying, for example, that $\pi_{s,t}$ is an $h$-hop-improving shortest path through $C$ in $G$ with respect to $G \setminus D$ to refer to a path $\pi_{s,t}$ that is in $G$ and has weight at most equal to the shortest path between $s$ and $t$ of hop $h$ that contains a vertex in $C$ in $G \setminus D$.

In this paper, we often use a black box a result by Zwick \cite{Zwick02} that extends $h$-hop-improving shortest paths in $G \setminus D$ to improving shortest-paths. Since the lemma is implicit in \cite{Zwick02} we provide an implementation of the algorithm and a proof of correctness that can be found in appendix \ref{sec:proofLemmaExtDist}.

\begin{lemma}[see \cite{Zwick02,AbrahamCK17}] \label{lemma:ExtendDists}
Given a collection $\Pi$ of the $h$-hop-improving shortest paths for all pairs $(s,t) \in V^2$ in $G \setminus D$, then there exists a procedure \ExtendDists that returns improving shortest paths for all pairs $(s,t) \in V^2$ in time $O(n^3 \log n/ h + n^2 \log^2 n)$.
\end{lemma}

\section{The Framework}

In this section, we describe the fundamental approach that we use for our data structures. We then refine this approach in the next section to obtain our new data structures. We start by stating a classic reduction that was used in all existing approaches.

\begin{lemma}[see \cite{henzinger2001maintaining, thorup2005worst, AbrahamCK17}] \label{lma:reductionToDecr}
Given a data structure on $G$ that supports a batch deletion of up to $2\Delta$ vertices $D \subseteq V$ from $G$ such that the data structure computes for each $(s,t) \in V^2$ a shortest path $\pi_{s,t}$ in $G \setminus D$, and can return the $k$ first edges on $\pi_{s,t}$ in time $O(k)$ time. Then, if the preprocessing time is $t_{pre}$ and the batch deletion worst-case time is $t_{del}$, there exists a fully dynamic APSP algorithm with $O(t_{pre}/\Delta + t_{del} + \Delta n^2)$ worst-case update time.
\end{lemma}

This lemma reduces the problem to finding a data structure with good preprocessing time that can handle batch deletions. To get some intuition on how the reduction stated above works, note that vertex insertions can be solved very efficiently. 

\begin{lemma}[implied by Johnson's algorithm, see for example \cite{cormen2009introduction}]
\label{lma:floydWarshall}
Given a graph $G \setminus C$, where $C \subset V$ of size $\Delta$, and given all-pairs shortest paths $G \setminus C$, we can compute all-pairs shortest paths in $G$ in time $O(\Delta n^2)$.
\end{lemma}

Therefore, the algorithm can reduce the problem to a decremental problem and using the shortest paths in $G \setminus (D \cup C)$ insert a batch of vertices $C \setminus D$ after each update. When $C$ becomes of size larger than $\Delta$ (i.e. after at least $\Delta$ updates), we recompute the decremental data structure. Using standard deamortization techniques for rebuilding the data structure, the preprocessing time can be split into small chunks that are processed at each update and therefore we obtain a worst-case guarantee for each update.

\subsection{A Batch Deletion Data Structure with Efficient Preprocessing Time}

In the following, we present a procedure \Prepro , given in Algorithm \ref{alg:detpreproc}, that is invoked with parameters $G$ and integers $h > 0$ and $\tau\geq 2n^2$ to compute paths $\pi^i_{s,t}$ for each tuple $(s,t) \in V^2$ and for every $i \in [0, i_h]$ with $i_h = \lceil\log_{3/2} h\rceil$. Our goal is to use these paths in the batch deletion to recompute all-pairs shortest paths. 

\begin{algorithm}
\caption{$\Prepro$}
\label{alg:detpreproc}
\KwIn{A graph $G=(V,E)$, a positive integer $h > 0$ determining the maximum hop and an integer $\tau \geq 2n^2$ regulating the congestion.}
\KwOut{A tuple $(C, \{\textsc{Congestion}(v)\}_{v \in V},\{\pi^i_{s,t}\}_{s,t\in V,i\in\{0,\ldots,i_h\}})$ with the properties of Lemma~\ref{lemma:PebbleSum}.}
\BlankLine
$C \gets \emptyset$\;
\lForEach{$v \in V$}{$\textsc{Congestion}(v) \gets 0$}
\lForEach{$s,t\in V, i \in [0, i_h]$}{$\pi^i_{s,t} \gets \nil$}

$X \gets V$  \;
\While(\label{lne:detWhileLoop}){$X \neq \emptyset$}{
    Remove an arbitrary root $s$ from $X$\label{lne:detPickArbitrary}\; 
    
    \ForEach(\label{lne:forloopdet}){$i \in [0, i_h]$}{
    $\{\pi^i_{s,t}\}_{t\in V} \gets \textsc{BellmanFord}(s, G[V \setminus C], h_i)$ \label{lne:computeShortestPathsPrepro}\;
    
    \ForEach(\label{lne:detForeachBegin}){$t \in V, u \in \pi^i_{s,t}$}{
            $\textsc{Congestion}(u) \gets \textsc{Congestion}(u) + \lceil n / h_i \rceil$\;
        }
        $C \gets \{v \in V| \textsc{Congestion}(v) > \tau/2\}$\;\label{lne:detForeachEnd}
    }
}
\end{algorithm}

This procedure maintains \textit{congestion} values $\textsc{Congestion}(u)$ for each $u\in V$. These counters are initially $0$. Let $i_h = \lceil\log_{3/2} h\rceil$ and let $h_i = (3/2)^i$, for $i = 0,\ldots,i_h$ throughout the rest of the article. For each such $i$, $h_i$-hop-restricted shortest paths $\pi^i_{s,t}$ in $G$ are computed from roots $s$ to all $t \in V$ where roots are considered in an arbitrary order.

For each $u \in V$, whenever an $h_i$-hop-restricted path $\pi^i_{s,t}$ is found that passes through $u$, $\textsc{Congestion}(u)$ is increased by $\lceil n/h_i\rceil$. Hence, paths of long hop congest vertices on them less than small hop paths; this is key to getting our update time improvement as it helps us to keep the amount of congestion at $O(n)$ for a path of \emph{any} hop (as opposed to existing techniques which can only bound the cost at $O(nh)$). Once a {congestion} value $\textsc{Congestion}(u)$ increases beyond threshold value $\tau/2$, $u$ is removed from the vertex set. More precisely, a growing set $C$ keeps track of the set of vertices whose {congestion} value is above $\tau/2$ and all hop-restricted paths are computed in the shrinking graph $G[V\setminus C]$.

\begin{lemma}\label{lemma:DetPreproc}
The procedure $\Prepro$ can be implemented to run in $O(n^3h)$ time.
\end{lemma}
\begin{proof}
In each iteration of the for-loop in line \ref{lne:forloopdet}, computing $h_i$-hop-restricted shortest paths from source $s$ to all $t\in V\setminus C$ can be done with $h_i$ iterations of Bellman-Ford from $s$ in $G[V\setminus C]$ in time $O(n^2h_i)$ time. It is straight-forward to see that this dominates the cost incurred by the accounting in lines \ref{lne:detForeachBegin} to \ref{lne:detForeachEnd}. From this and from a geometric sums argument, it follows that the total running time over all sources $s$ is $O(n\cdot \sum_{i=0}^{\lceil \log_{3/2} h\rceil}  n^2 (3/2)^{i}) = O(n^3 h)$. The lemma now follows.
\end{proof}

To bound the time for updates in the next subsection, we need the following lemma. 

\begin{lemma}\label{lemma:PebbleSum}
At termination of $\Prepro$, the algorithm ensures that
\begin{enumerate}
    \item $\forall v \in V$: $\textsc{Congestion}(v) \leq \tau$,
    \item $\sum_{v\in V}\textsc{Congestion}(v) = O(n^3\log h)$, and
    \item $|C| = O(n^3 \log h/ \tau)$.
    \item Each computed path $\pi^i_{s,t}$ is a \emph{$h_i$-hop-improving shortest path} in $G$ with regard to $G \setminus C$.
\end{enumerate}
\end{lemma}
\begin{proof}
We first observe that we maintain the loop invariant for the while-loop in line \ref{lne:detWhileLoop} that the congestion of any vertex not in $C$ is at most $\tau/2$. This is true since initially the congestion of all vertices is $0$ and at the end of each iteration, we explicitly remove vertices with congestion larger than $\tau/2$ from $C$. To prove property 1, it therefore suffices to show that during an iteration of the while-loop in line \ref{lne:detWhileLoop}, the congestion of any vertex is increased by at most $\tau/2$. To see this, observe that there are at most $n$ paths under consideration in each iteration of the while loop. Every vertex $u$ has its congestion increased by $\lceil n / h_i \rceil \leq n$ for each such path it belongs to. Therefore, we add at most $n^2 \leq \tau /2$ congestion to any vertex $u$ during an iteration of the while-loop.

To see property 2, define $\Phi = \sum_{v\in V}\textsc{Congestion}(v)$. Initially, $\Phi = 0$. Observe that during an iteration of the while-loop in line \ref{lne:detWhileLoop}, we have at most $n$ paths of hop up to $h_i$. Thus at most $h_i + 1$ vertices increase their congestion due to a path by $\lceil n / h_i \rceil$ and so each such path increases $\Phi$ by at most $O(n)$. Thus each while-loop iteration adds at most $O(n^2)$ to $\Phi$ and since we execute the while-loop exactly $n \lceil \log_{3/2} h \rceil$ times, the final value of $\Phi$ is $O(n^3\log h)$. 

Property 3 follows since each vertex $c \in C$ has congestion at least $\tau/2$, implying that there can be at most $2\Phi / \tau = O(n^3 \log h/\tau)$ vertices in $|C|$. Property 4 follows from the analysis of Bellman-Ford.
\end{proof}

The space-efficiency is straight-forward to analyze since each pair $(s,t) \in V^2$ requires one path $\pi^i_{s,t}$ to be stored for each $i \in [0, i_h]$, storing the shortest paths explicitly requires space $O(\sum_{i=0}^{\lceil \log_{3/2} h\rceil} n^2 (3/2)^{i}) = O(n^2 h)$. We defer the description of a more space-efficient data structure with the same guarantees until Section \ref{subsec:spaceEfficientDet}.

\subsection{Handling Deletions}
\label{subsec:handlingDeletions}

In this section, we use the data structure computed by $\Prepro$ with $C$ being again the set of congested vertices, and show how to use this data structure to handle a batch $D \subseteq V$ of at most $2\Delta$ deletions, i.e. we show how to efficiently compute all-pairs shortest paths in $G \setminus D$. Our update procedure proceeds in multiple phases $1, \dots, i_h$. Throughout the procedure, we enforce the following invariant.

\begin{invariant}
\label{inv:phaseInvariant}
For every $(s,t) \in V^2$ where $\Pi_{s,t}$ is the collection of shortest-path from $s$ to $t$ in $G \setminus D$, and 
\begin{itemize}
    \item no $\pi_{s,t} \in \Pi_{s,t}$ contains a vertex in $C$, and 
    \item there is some shortest-path $\pi_{s,t} \in \Pi_{s,t}$ of hop at most $h_i$.
\end{itemize}
Then, after the execution of the $i^{th}$ phase, we have that $\pi^i_{s,t}$ is a shortest-path in $\Pi_{s,t}$ of minimal hop.
\end{invariant}

Before we describe how to enforce the invariant, observe that the invariant implies that after we finished phase $i_h$, we have for each pair $(s,t)$ a $h$-hop-improving shortest path in $G \setminus D$ which can then be extended using procedure $\textsc{DetExtDistances}(\{N_i\}_i, h)$ as described in Lemma \ref{lemma:ExtendDists} to give all-pairs shortest paths in $G \setminus D$, as required.

\begin{algorithm}
\caption{$\textsc{Delete}(D, h)$}
\label{alg:delete}
\DontPrintSemicolon
\BlankLine
\ForEach(\label{lne:deleteInit}){$\pi^0_{s,t} \cap D \neq \emptyset$}{
    $\pi^0_{s,t} \gets \nil$
}

\For(\label{lne:deleteMain}){$i \gets 1$ to $ i_h$}{
    \ForEach{$s \in V$}{
        \If{$h_i > 3$}{
            Compute an integer $\textsc{Rad}^{i}(s) \in (\frac 1 3 h_i,\frac 2 3 h_i)$ that minimizes the size of  $\textsc{Separator}^i(s) = \{x \in V \setminus \{D \cup C\} \;\vert\; |\pi^{i-1}_{s,x}| = \textsc{Rad}^{i}(s)\}$ \label{lne:deleteSeparator}
        }\Else{$\textsc{Separator}^i(s) \gets V$}
    }

    \ForEach(\label{lne:foreachDeleteMain}){$\pi^i_{s,t} \cap D \neq \emptyset$}{
        Let $x$ be any vertex in $\textsc{Separator}^i(s) \cup \{t\}$ such that for any vertex $y \in \textsc{Separator}^i(s) \cup \{t\}$, either $w(\pi^{i-1}_{s,x} \circ \pi^{i-1}_{x,t}) < w(\pi^{i-1}_{s,y} \circ \pi^{i-1}_{y,t})$ or $w(\pi^{i-1}_{s,x} \circ \pi^{i-1}_{x,t}) = w(\pi^{i-1}_{s,y} \circ \pi^{i-1}_{y,t})$ and $|\pi^{i-1}_{s,x} \circ \pi^{i-1}_{x,t}| \leq |\pi^{i-1}_{s,y} \circ \pi^{i-1}_{y,t}|$.  \label{lne:deleteNewPath}\;
        $\pi^i_{s,t} \gets \pi^{i-1}_{s,x} \circ \pi^{i-1}_{x,t}$
    }
}
\Return $\textsc{DetExtDistances}(\{ \pi^{i_h}_{s,t} \}_{(s,t) \in V^2}, h_i)$
\end{algorithm}

Let us now describe how to implement the execution of a phase which is also depicted in Algorithm \ref{alg:delete}. Initially, we change all precomputed paths $\pi^0_{s,t}$ with $s$ or $t$ in $D$ to the empty path $\bot$. Clearly, this enforces Invariant \ref{inv:phaseInvariant} and can be implemented in $O(n^2)$ time.

In the $i^{th}$ phase (for $i > 0$), we start by computing for each vertex $s \in V$, a hitting set of all $h_{i-1}$-hop-improving shortest paths starting in $s$. We take the separator set $\textsc{Separator}^i(s)$ such that in particular each (real) shortest path from $s$ of length at least $\frac 2 3 h_i$ contains at least one vertex in $\textsc{Separator}^i(s)$ that is at distance $\textsc{Rad}^i(s)$ from $s$. Here $\textsc{Rad}^i(s)$ is chosen to be between the $(\frac 1 3 h_i)^{th}$ and $(\frac 2 3 h_i)^{th}$ vertex on each path (with exception for very small $h_i$ where we chose the separator to be the entire vertex set). Since there are $\Theta(h_i)$ layers to chose $\textsc{Rad}^i(s)$ from, and the layers partition the vertex set $V$, we obtain that $\textsc{Separator}^i(s)$ is of size $O(n / h_i)$ by the pigeonhole principle. Finally, to fix any precomputed $h_i$-hop-improving shortest path $\pi^i_{s,t}$ that is no longer in $G \setminus D$, we check the paths $\pi^{i-1}_{s,x} \circ \pi^{i-1}_{x,t}$ for each $x \in \textsc{Separator}^i(s) \cup \{t\}$ and take a path of minimal weight (and among those of minimal hop). We point out that this path is either the concatenation of two $h_{i-1}$-hop-improving shortest paths, or the path $\pi^{i-1}_{s,t}$. This completes the description of our update algorithm.

\begin{lemma}
The Invariant \ref{inv:phaseInvariant} is enforced throughout the entire execution of procedure $\textsc{Delete}(D, h)$.
\end{lemma}
\begin{proof}
Before the loop starts the invariant is clearly enforced since for $i = 0$, we either have an edge between two points or not. Let us therefore take the inductive step for $i > 0$ and let us focus on some path $\pi^i_{s,t}$. Clearly, if $\pi^i_{s,t}$ contains no vertex in $D$, it is still $h_i$-hop-improving in $G \setminus D$ and therefore no action is required. Otherwise, let $\widehat{\pi_{s,t}}$ be some $s$-to-$t$ shortest-path in $\Pi_{s,t}$ of minimal hop (we assume that no shortest path intersects $C$). Clearly, if $\widehat{\pi_{s,t}}$ has $|\widehat{\pi_{s,t}}| \leq h_{i-1}$, then we have by the induction hypothesis that $\pi^{i-1}_{s,t}$ is a shortest-path from $s$-to-$t$ of minimal hop, and thus if $x = t$, we obtain $\pi^i_{s,t} = \pi^{i-1}_{s,t}$. If $x \neq t$, then the path $\pi^i_{s,t}$ is set to another shortest-path of minimal hop by the way we choose $x$.

It remains to consider the case where $h_i \leq |\widehat{\pi_{s,t}}| > h_{i-1}$. Then, let $\hat{x}$ be the $\textsc{Rad}^t(s)^{th}$ vertex on $\widehat{\pi_{s,t}}$ (which exists since $\textsc{Rad}^t(s) < \frac{2}{3}h_i = h_{i-1} < |\widehat{\pi_{s,t}}|$). Then, observe that $\hat{x} \in \textsc{Separator}^i(s) \cup \{t\}$ since by induction hypothesis every $s$-to-$\hat{x}$ shortest-path $\pi_{s,\hat{x}}$ of minimal hop has exactly $\textsc{Rad}^t(s)$ hops and $\pi^{i-1}_{s,\hat{x}}$ is chosen among these paths by induction hypothesis (also none of these paths intersects $C$ since no $s$-to-$t$ shortest-path does). Similarly, we have that $\pi^{i-1}_{\hat{x},t}$ is a shortest path of minimal hop from $\hat{x}$ to $t$. This implies that the path $\pi^{i-1}_{s,\hat{x}} \circ \pi^{i-1}_{\hat{x},t}$ is a shortest $s$-to-$t$ path of minimal hop. Since $x$ has $\hat{x}$ among its choices, we thus have that $\pi^i_{s,t} = \pi^{i-1}_{s,{x}} \circ \pi^{i-1}_{{x},t}$ is a shortest $s$-to-$t$ path of minimal hop. The lemma follows. 
\end{proof}

\begin{lemma}
Given a data structure that satisfies the properties listed in Lemma \ref{lemma:PebbleSum} with congestion threshold $\tau$ and a set of congested vertices $C$, there exists an algorithm that computes all-pairs shortest paths in $G \setminus D$ and returns the corresponding distance matrix in time $O(|D| \tau + |C| n^2 + n^3 \log n /h)$.
\end{lemma}
\begin{proof}
By Invariant \ref{inv:phaseInvariant}, we obtain all shortest paths of hop at most $h$ for pairs that have no shortest path through $C$. Thus, it is straightforward to adapt the procedure described in Lemma \ref{lma:floydWarshall} to return in $O(|C| n^2)$ time a collection of $h$-hop-improving shortest paths in $G \setminus D$. Finally, the Lemma \ref{lemma:ExtendDists} can be applied to recover in $O(n^3 \log n/h)$ time the shortest paths in $G \setminus D$, as required.

It remains to analyze the running time of Algorithm  \ref{alg:delete}. We note that each phase requires us to compute a separator for each vertex in $V$. Since returning the first $h_i$ edges of each path $\pi^{i-1}_{s,x}$ requires time $O(h_i)$ since we represent paths explicitly, the time required to compute a single separator in phase $i$ is at most $O(n h_i)$. Thus, the overall time to compute all separators can be bound by $O(n^2 h)$ (using a geometric sum argument for the different phases).

To bound the time spend in the foreach-loop in line \ref{lne:foreachDeleteMain}, observe that we iterate only over paths that contain a vertex in $D$ which can be detected in linear time. Observe that if a vertex $v$ in $D$ is on a path $\pi_{s,t}^i$, then the path contributed $\lceil n/h_i \rceil$ credits to the congestion of $v$ in the preprocessing procedure. Since the separator of $s$ at phase $i$ has size $O(n/h_i)$ by the arguments mentioned above, we have that the iteration to recover path $\pi_{s,t}^i$ requires time $O(n/h_i)$ (that is since checking the weight of each path and concatenation can both be implemented in constant time). Since each vertex $v \in D$ has total congestion at most $\tau$ by Lemma \ref{lemma:PebbleSum}, we can bound the total running time of the algorithm by $O(|D| \tau + n^2 h)$.
\end{proof}

Choosing $\tau = n^{2+1/4} \sqrt{\log n}, h = n^{1/4} \sqrt{\log n}$ and $\Delta = n^{1/2}$ in Lemma \ref{lma:reductionToDecr}, we obtain the following corollary. 

\begin{corollary}
Let $G$ be an $n$-vertex directed edge-weighted graph undergoing vertex insertions and deletions. Then there exists a deterministic data structure which can maintain distances in $G$ between all pairs of vertices in worst-case update time $O(n^{2+3/4} \sqrt{\log n})$.
\end{corollary}

\subsection{Batch Deletion Data Structure for Unweighted Graphs}

We point out that for unweighted graphs, we can replace the Bellman-Ford procedure by a simple Breath-First-Search procedure (see for example \cite{cormen2009introduction}) which improves the running time from $O(n^2 h)$ to $O(n^2)$. This was also exploited before in \cite{AbrahamCK17}. 

\begin{corollary}
Let $G$ be an $n$-vertex directed edge-weighted graph undergoing vertex insertions and deletions. Then there exists a deterministic data structure which can maintain distances in $G$ between all pairs of vertices in worst-case update time $O(n^{2+2/3} (\log n)^{2/3})$.
\end{corollary}

In the following sections, we will not explicitly point out that the Bellman-Ford procedure can be replaced by BFS but simply state the improved bound. 

\section{Efficient Data Structures}
 
We now describe how to use the general strategy described in the previous section and describe the necessary changes to obtain efficient data structures.

\subsection{A Faster Deterministic Algorithm}

To obtain a faster algorithm, we mix our framework with the following result from by Abraham, Chechik and Forster~\cite{AbrahamCK17}. It is not explicitly stated in their paper but follows immediately by replacing their randomly sampled vertex set by an arbitrary vertex set. Informally, the data structure takes a decremental graph $G$ and a set $C \subseteq V$ of vertices and maintains for all vertices $v \in V$, the shortest-path through some vertex in $C$.

\begin{lemma}\label{lemma:CenterDS}
Given an edge-weighted directed graph $G = (V,E)$, a set $C \subseteq V$ and a hop bound $h$. Then there exists a deterministic data structure that supports the operations:
\begin{itemize}
    \item $\textsc{ACKPreprocessing}(G, C, \Delta', h)$: Initializes the data structure with the given parameters and returns a pointer to the data structure.
    \item $\textsc{ACKDelete}(D)$: assuming $D \subseteq V$, returns for each $(s,t) \in V^2$, a $h$-hop-improving shortest path $\pi_{s,t}$ through some vertex in $C \setminus D$ in $G \setminus D$ (with respect to $G \setminus D$).
\end{itemize}
The operation $\textsc{ACKPreprocessing}(G,C,h)$ runs in $O(|C|n^2 h)$ time and each operation $\textsc{ACKDelete}(D)$ runs in $O(|D| n^2 h \log n + |C|nh)$ time.
\end{lemma}

It is now straight-forward to obtain a new batch deletion data structure that combines these two data structures. Intuitively, we exploit the strengths of both algorithms by setting the $\tau$-threshold of the algorithm introduced in previous section slightly lower which increases the size of the set $C$ of congested vertices but improves the running time of the data structure to maintain shortest-paths that do not contain any vertices in $C$. Since $C$ is precomputed, we then use the data structure described above to obtain the shortest-paths through some vertex in $C$. Let us now give a more formal description.

To initialize the new data structure, we invoke algorithm \ref{alg:detpreproc} with parameters $\tau$ and $h$ to be fixed later. The algorithm gives a data structure $\mathcal{D}_1$ and a set $C$ is of size $O(n \log n /\tau)$. We then handle $2\Delta$ updates as follows: At initialization and every $\Delta'$ updates, we compute a data structure $\mathcal{D}_2$ by invoking the preprocessing algorithm in Lemma \ref{lemma:CenterDS} with parameters $C$ and $h$. We later chose $\tau$ larger than in the last section which implies that we can increase $\Delta$, and take care of the shortest paths through $C$ by recomputing $\mathcal{D}_2$ more often, i.e. we set $\Delta' \ll \Delta$. Since the preprocessing time of $\mathcal{D}_2$ is smaller, this can be balanced efficiently such that both have small batch update time at all times.

For each update, we let $D_1$ be the batch of deletions since $\mathcal{D}_1$ was initialized and $D_2$ the batch of deletions since $\mathcal{D}_2$ was initialized. We then invoke $\mathcal{D}_1.\textsc{Delete}(D_1)$ and $\mathcal{D}_2.\textsc{ACKDelete}(D_2)$ and combine the results in a straight-forward manner. This concludes the algorithm.

Using the reduction \ref{lma:reductionToDecr}, and using that $|D_1| \leq \Delta$ and $D_2 \leq \Delta'$, we obtain an algorithm with worst-case running time  
\[
  O(n^3h/\Delta + \Delta n^2 + n^3(\log n)^2/h + \Delta\tau + \Delta' hn^2\log n + n^4h(\log h)/\tau + n^5 h(\log h)/(\tau\Delta'))
\]
which is optimized by setting $\tau = n^{7/3}h^{2/3}(\log n\log h)^{1/3}/\Delta^{2/3}$, $h = n^{2/7}(\log n)^{6/7}$, $\Delta = \sqrt{n}h^{1/4}/(\log n\log h)^{1/4}$, and $\Delta'=\sqrt{n^3(\log h)/(\tau\log n)}$.

\begin{corollary}
Let $G$ be an $n$-vertex directed edge-weighted graph undergoing vertex insertions and deletions. Then there exists a deterministic data structure which can maintain distances in $G$ between all pairs of vertices in worst-case update time $O(n^{19/7}(\log n)^{8/7})$. If the graph is unweighted, the running time can be improved to $O(n^{2.6}\log n)$.
\end{corollary}

\subsection{A Simple and Space-Efficient Deterministic Data Structure}
\label{subsec:spaceEfficientDet}

In order to reduce space, we replace the procedure  $\textsc{BellmanFord}(s, G[V \setminus C], h_i)$ in the preprocessing at line \ref{lne:computeShortestPathsPrepro}  by procedure $\textsc{BellmanFordSpaceEfficient}(s, G[V \setminus C], h_i)$ that is depicted in algorithm \ref{alg:detBellmanFord}. Unlike the Bellman-Ford algorithm, our algorithm does not return the $h_i$-restricted shortest paths but instead returns $h_i$-improving shortest paths of length at most $O(h_i)$. Using that the length of each $h_i$-improving shortest paths is $O(h_i)$, it can be verified that the proof of lemma \ref{lemma:PebbleSum} still holds under these conditions. Moreover, the information computed by $\textsc{BellmanFordSpaceEfficient}(s, H, h_i)$ can be efficiently stored in $\tilde{O}(n)$ space.

\begin{algorithm}
\caption{$\textsc{BellmanFordSpaceEfficient}(s, H, h_i = (3/2)^i)$}
\label{alg:detBellmanFord}
\KwIn{A graph $H$ and source $s \in H(V)$, an integer $h_i = (3/2)^i \geq 1$.}
\KwOut{The algorithm returns a set of $h_i$-hop-improving shortest paths $\{\pi^i_{s,t}\}_{t\in V}$ each of length $O(h_i)$ that can be represented in space $\tilde{O}(n)$. }

\BlankLine

$H' \gets H; j \gets i$\;
\For{$j \gets i$ down to $0$}{
    $\{\tau^j_{s,t}\}_{t\in V} \gets \textsc{BellmanFord}(s, H^j, h_j)$\;
    
    \If{$h_j > 3$}{
        Compute an integer $\textsc{Rad}^{j}(s) \in (\frac 1 3 h_j,\frac 2 3 h_j)$ that minimizes the size of  $\textsc{Separator}^j(s) = \{x \in V \vert |\pi^{j}_{s,x}| = \textsc{Rad}^{j}(s)\}$ \label{lne:deleteSeparator}
    }\Else{$\textsc{Separator}^j(s) \gets V$}
    
    \ForEach(\label{lne:storePaths}){$x \in \textsc{Separator}^j(s)$}{
        Store path $\tau^{j}_{s,x}$\;
        $w'(s,t) \gets w(\tau^{j}_{s,x})$
    }
}
\end{algorithm}

The algorithm runs in iterations executed by the for-loop where the index $j$ is initially set to $i$ and decreased after every iteration until it is $0$. In each iteration, we compute the hop-$h_j$-restricted shortest paths $\{\tau^j_{s,t}\}_{t\in V}$ on the graph $H'$. For the sake of analysis, we let $H^j$ be the graph of $H'$ at the start of iteration $j$. After computing the paths on $H^j$, we compute a separator set $\textsc{Separator}^j(s)$ that contains all vertices whose shortest path from $s$ has length $\textsc{Rad}^j(s)$ which is taken to be strictly between $\frac 1 3 h_j$ and $\frac 2 3 h_j$ (except for very small $h_j$ where we chose the separator set to be $V$). 

In the foreach-loop in line \ref{lne:storePaths}, we store for every vertex $x$ in the hitting set $\textsc{Separator}^j(s)$, the $h_j$-hop-restricted shortest path $\tau^j_{s,x}$. If the first edge on the path represents a subpath from a higher level, we add a pointer to the subpath. Then, we update $H^j$ by setting the weight of the edge from $s$ to $x$ to the weight of $\tau^j_{s,x}$. Observe that after the foreach-loop finished, all paths $\tau^j_{s,t}$, for any $t \in V(H)$, can be mapped to a path in $H'$ of same weight and of length at most $h_{j-1} = \frac 2 3 h_j$ and observe that this graph is graph $H^{j-1}$. 

Finally, we store the paths $\pi^i_{s,t}$ by a pointer to $\tau^0_{s,t}$. Observe that each path $\tau^0_{s,t}$ might then be unpacked to an $h_i$-improving shortest path in $H$ by replacing the first edge on a path by the corresponding subpath on a higher level.

\begin{lemma}\label{lma:spaceEfficientBellman}
The procedure $\textsc{BellmanFordSpaceEfficient}(s, H, h_i)$ computes a collection of $h_i$-hop-improving shortest paths $\{\pi^i_{s,t}\}_{t\in V}$ from source $s$ where each path is of length at most $O(h_i)$ and provides a $O(n \log h_i)$ sized data structure such that:
\begin{enumerate}
    \item Each path $\pi^i_{s,t}$ can be extracted from the data structure in time $O(|\pi^i_{s,t}|)$, and
    \item $\forall u \in V$, we can identify all paths $\Pi_u = \{\pi^i_{s,t} | u \in \pi^i_{s,t} \}$ that contain $u$ in time $O(|\Pi_u|)$.  
\end{enumerate}
The procedure takes time $O(n^2 h_i)$.
\end{lemma}
\begin{proof}
We argued above that every hop $h_j$-restricted shortest path in $H^j$ can be mapped to a $h_{j-1}$-restricted shortest path in $H^{j-1}$. Thus, computing the $h_{j-1}$-restricted shortest path using Bellman-Ford on $H^{j-1}$ returns $h_j$-hop-improving shortest paths. By a simple inductive argument, it follows that every shortest path $\pi_{s,t}^j$ for any $j$ is $h_i$-hop-improving in regard to $H$.

To see that every path $\pi^i_{s,t}$ is of length $O(h_i)$ observe that on level $j$, we add at most $h_j-1$ new edges to the path since the only subpaths that we replace by shortcuts are $s$ to $t$ paths. Thus the final path corresponds to a path of length $O(\sum_{j = 0}^{i} h_j) = O(h_i)$.

To see that the data structure requires only $\tilde{O}(n)$ space, observe that at iteration $j$, each path $\tau^j_{s,t}$ computed on $H^j$ consists of at most $O(h_j)$ edges that need to be stored explicitly and a pointer to a higher level subpath corresponding to the first edge of $\tau^j_{s,t}$. Since $|\textsc{Separator}^j(s)| = O(n /h_j)$, and we only store paths $\tau^j_{s,x}$ to each $x \in \textsc{Separator}^j(s)$, we therefore only require space $O(\sum_{j=0}^i h_j \cdot n /h_j) = O(n \log h_i)$

We can further implement the pointers for the subpath corresponding the first edge on a path $\pi^i_{s,t}$ to point to the next higher level where the subpath is non-trivial (i.e. not itself an edge). Thus following a pointer we can ensure to add at least one additional edge to the path, and therefore we can extract the path in time $O(|\tau^i_{s,t}|)$. Making pointers of the structure bidirectional, we can also find all paths $\pi_{s,t}^i$ containing a vertex $u$ in linear time. The overall running time is dominated by running Bellman-Ford, which takes $O(\sum_{j=0}^i n^2 h_j) = O(n^2h_i)$ time.
\end{proof}

The lemma gives a straight-forward way to verify that Lemma \ref{lemma:DetPreproc} and \ref{lemma:PebbleSum} hold even by using the relaxed Bellman-Ford procedure. The corollary below follows.

\begin{corollary}
Let $G$ be an $n$-vertex directed edge-weighted graph undergoing vertex insertions and deletions. Then there exists a deterministic data structure which can maintain distances in $G$ between all pairs of vertices in worst-case update time $O(n^{2+3/4} (\log n)^{2/3})$ using space $\tilde{O}(n^2)$. If the graph is unweighted, the running time can be improved to $O(n^{2+2/3} (\log n)^{2/3})$.
\end{corollary}

\subsection{A Space-Efficient and More Robust Las-Vegas Algorithm}

In this section, we present a simple randomized procedure that allows to refine the approach of our framework. On a high-level, we set the congestion threshold for each vertex quite small (very close to $n^2$). Whilst this implies that our set of congested vertices $C$ is quite large, we ensured that we have the paths in $G \setminus C$ covered for many deletions. We then try to fine recursively all paths through vertices in $C$ with slightly larger congestion threshold. By shrinking the set $C$ in each iteration, we speed-up the proprecessing procedure and therefore we can re-compute the data structure more often. We point out that even though our layering process again gives an efficient data structure to maintain paths that go through vertices in $C$, it does not rely on the techniques by Abraham, Chechik and Krinninger \cite{AbrahamCK17}.

\begin{algorithm}
\caption{$\textsc{RandPreprocessing}(G, C_{in}, \tau, h)$}
\label{alg:randpreproc}
\KwIn{A graph $G=(V,E)$, a positive integer $h > 0$ determining the maximum hop and an integer $\tau \geq 2n^2$ regulating the congestion.}
\KwOut{A tuple $(C_{out}, \{\textsc{Congestion}(v)\}_{v \in V},\{\pi^i_{s,t}\}_{s,t\in V,i\in\{0,\ldots,i_h\}})$ with the properties of Lemma~\ref{lemma:PebbleSumRand}.}
\BlankLine
$C_{out} \gets \emptyset$\;
\lForEach{$v \in V$}{$\textsc{Congestion}(v) \gets 0$}
\lForEach{$s,t\in V, i \in [0, i_h]$}{$\pi^i_{s,t} \gets \nil$}

$X \gets C_{in}$  \;
\While(\label{lne:randWhileLoop}){$X \neq \emptyset$}{
    Remove a center $c$ uniformly at random from $X$\label{lne:randPickRandom}\; 
    
    \ForEach(\label{lne:forlooprand}){$i \in [0, i_h]$}{
    $\{\pi^i_{c,t}\}_{t\in V} \gets \textsc{BellmanFordSpaceEfficient}(c, G[V \setminus C_{out}], h_i)$ \label{lne:randComputeShortestPathsPrepro}\;
    $\{\pi^i_{s,c}\}_{s\in V} \gets \textsc{BellmanFordSpaceEfficient}(c, \overleftarrow{ G[V \setminus C_{out}]}, h_i)$ \;
        
        \ForEach(\label{lne:foreachImprovedPair}){$s,t \in V$ with $w(\pi_{s,c}^i \circ \pi_{c,t}^i) < w(\pi_{s,t}^i) $}{
            $\pi_{s,t}^i \gets \pi_{s,c}^i \circ \pi_{c,t}^i$\;
            \ForEach(\label{lne:randForeach2Begin}){$u \in \pi^i_{s,t}$}{
                $\textsc{Congestion}(u) \gets \textsc{Congestion}(u) + \lceil n / h_i \rceil$\label{lne:addPairCongestion}\;
            }
            \If(\label{lne:ifRandCongestionExceeded}){$C_{out} \neq \{v \in V| \textsc{Congestion}(v) > \tau/2\}$}{
                    $C_{out} \gets \{v \in V| \textsc{Congestion}(v) > \tau/2\}$\;\label{lne:randForeachEnd}
                    $\{\pi^i_{c,t}\}_{t\in V} \gets \textsc{BellmanFordSpaceEfficient}(c, G[V \setminus C_{out}], h_i)$ \label{lne:ReComputeShortestPathsPrepro}\;
                    $\{\pi^i_{s,c}\}_{s\in V} \gets \textsc{BellmanFordSpaceEfficient}(c, \overleftarrow{ G[V \setminus C_{out}]}, h_i)$ \;
                }
        }
    }
}
\end{algorithm}

We start by presenting an adapted version of the preprocessing algorithm \ref{alg:detpreproc} that is depicted in algorithm \ref{alg:randpreproc}. The new algorithm takes a set $C_{in}$ of vertices and the goal of the procedure is to produce $h$-hop-improving shortest paths through the vertices in $C_{in}$ in the graph $G[V \setminus C_{out}]$ where $C_{out}$ is a set of vertices that are congested over the course of the algorithm. Instead of taking vertices from $X$ in arbitrary order, we now sample a vertex $c$ uniformly at random in each iteration. We then compute $h_i$-hop-improving shortest paths from and to $c$ by invoking the adapted Bellman-Ford procedure on the original and the reversed graph.

We then test whether the concatenation of $\pi^i_{s,c} \circ \pi^i_{c,t}$ has lower weight than the current path from $s$ to $t$. If so, we add $\lceil n/h_i \rceil$ units of congestion to each vertex $u$ on the path $\pi^i_{s,c} \circ \pi^i_{c,t}$. In contrast to previous algorithms, if the congestion of one of the vertices $u$ exceeds $\tau$, we immediately remove $u$ from the graph and recompute the paths through the vertex $c$ in the new graph. Let us now analyze the algorithm. 

\begin{lemma}\label{lemma:PebbleSumRand}
At termination of $\textsc{RandPreprocessing}(G, C_{in}, \tau, h)$, the algorithm ensures that
\begin{enumerate}
    \item $\forall v \in V$: $\textsc{Congestion}(v) \leq \tau$,
    \item $\sum_{v\in V} \textsc{Congestion}(v) = O(n^3 (\log n)^3)$, and
    \item $|C_{out}| = O(n^3 (\log n)^3/ \tau)$.
    \item Each computed path $\pi^i_{s,t}$ is a \emph{$h_i$-hop-improving shortest path} in $G$ through a vertex in $C_{in}$ with regard to $G \setminus C_{out}$.
\end{enumerate}
The algorithm runs in time $O(|C_{in} \cup C_{out}|n^2 h)$ and the properties are satisfied with probability $1-n^{-c}$ for any constant $c > 0$.
\end{lemma}
\begin{proof}
We maintain again the loop invariant for the while-loop in line \ref{lne:randForeach2Begin} that the congestion of any vertex not in $C_{out}$ is at most $\tau/2$. It is further straight-forward to see that during a loop-iteration only one path is evaluated at a time and therefore a single vertex can only obtain an addition $\lceil n/h_i \rceil \leq 2n$ units of congestion. Since $2n \leq \tau/2$, we immediately get property 1.

Next, let us analyze the total amount of congestion added in line \ref{lne:addPairCongestion}. Observe therefore that we add $\lceil n / h_i \rceil$ congestion to at most $2(h_i + 1)$ vertices respectively each time a pair $(s,t) \in V \times V$ has its path $\pi^i_{s,t}$ (strictly) improved by a path $\pi^i_{s,c} \circ \pi^i_{c,t}$ through a randomly drawn center vertex $c \in C$. Further, observe that we pick the centers in random order in line \ref{lne:randPickRandom}. Now let us first assume that all paths to and from centers are computed in the graph $G$ (instead of being dependent on the current state of the set $C_{out}$). Let $c$ be the $j^{th}$ chosen center and let us analyze the probability that $\pi^i_{s,c} \circ \pi^i_{c,t}$ is better than the best center seen so far. Clearly, this probability is
\[
    Pr[c \text{ is better than the best center seen so far} ] = \frac{1}{j}.
\]
Thus, the expected amount of congestion $X^i_{s,t}$ added in line \ref{lne:addPairCongestion} for a pair $(s,t)$ in phase $i$ (under the condition that we compute Bellman-Ford always on $G$) is $E[X^i_{s,t}] = \sum_{j = 1}^{|C_{in}|}  \frac{1}{j} 2(h_i + 1)\lceil n / h_i \rceil = O(n H_n)$ where $H_n$ is the $n^{th}$ harmonic number. Since we compute Bellman-Ford on $G[V \setminus C_{out}]$ and vertex deletions might increase distances, we conclude that the random variable $Y^i_{s,t}$ that is the expected amount of congestion added in line \ref{lne:addPairCongestion} for a pair $(s,t)$ using graph $G[V \setminus C_{out}]$ is stochastically dominated by $X^i_{s,t}$. Hence $E[Y^i_{s,t}] \leq O(n \log n)$ and using Markov's inequality, we have with constant probability that $Y^i_{s,t}$ is at most $O(n \log n)$. Using the Chernoff-Bound, we obtain for any $(s,t) \in V^2, i \in [0, i_h]$
\[
    Pr[Y^i_{s,t} > c' * n (\log n)^2] = 1- n^{- (c + 2)}
\]
where $c'$ is a constant chosen to be large. Finally, summing over all these events (for each tuple $(s,t)$ and $i \in [0,i_h]$), a union bound implies that the probability that the total congestion exceeds $O(n^3 (\log n)^3)$ is at most $1 - n^{-c}$. Thus, with probability $1-n^{-c}$, the algorithm satisfies property 2 on termination.

Property 3 follows since each vertex $c \in C_{out}$ has congestion at least $\tau/2$, implying that there can be at most $2\Phi / \tau = O(n^3 \log h/\tau)$ vertices in $C_{out}$. Property 4 follows from Lemma \ref{lma:spaceEfficientBellman}. The running time is strictly dominated by running Bellman-Ford from and to each vertex in $C_{in}$ and at most $C_{out}$ recomputations in the if-case in \cref{lne:ifRandCongestionExceeded}.Thus, the total running time is $O(|C_{in} \cup C_{out}|n^2 h)$ using again a geometric sum argument and Lemma \ref{lma:spaceEfficientBellman}.
\end{proof}

It is straight-forward to verify that the update procedure from section \ref{subsec:handlingDeletions} can be used to recover pre-computed paths that are destroyed by batch deletions. The following lemma formalizes that we can maintain efficiently $h$-hop-improving shortest paths through the set $C_{in}$ in the graph $G \setminus D$ with regard to $G \setminus (D \cup C_{out})$.

\begin{lemma}
\label{lma:updateRandStructure}
Given a data structure that satisfies the properties listed in Lemma \ref{lemma:PebbleSumRand} with congestion threshold $\tau$, with set $C_{in} \subseteq V$ and with a set of congested vertices $C_{out}$, there exists an algorithm that computes for each tuple $(s,t)$ an improving shortest path $\pi_{s,t}$ through $C_{in}$ in $G \setminus D$ with regard to $G \setminus (D \cup C_{out})$ and returns the corresponding distance matrix in time $O(|D|\tau + n^3/h)$.
\end{lemma}

Using this new algorithm to preprocess, we maintain data structures $\mathcal{D}_0, \mathcal{D}_1, \dots , \mathcal{D}_{\lceil \lg h \rceil}$ where each data structure $\mathcal{D}_i$ is initialized by invoking procedure $\textsc{RandPreprocessing}(G, C_i, \tau_i, h)$ that returns a set $C_{i+1}$ of congested vertices (here $C_i$ takes the role of $C_{in}$ and $C_{i+1}$ the role of $C_{out}$). Initially, we set $C_0 = V$, and $\tau_i = c * (\log n)^3 * 2^i n^2$ where $c$ is chosen such that $C_{1}$ in Lemma \ref{lemma:PebbleSumRand} is of size at most $n/2$ which also stipulates that $|C_i| \leq n/2^i$ for all $i$. Computing the data structure $\mathcal{D}_i$ once data structure $\mathcal{D}_{i-1}$ is computed is then well-defined. We finish the initialization with the data structures $\mathcal{D}_0, \mathcal{D}_1, \dots , \mathcal{D}_{\lceil \lg h \rceil}$ and a final set $C_{\lceil \lg h \rceil + 1}$ of size $O(n (\log n)^3/h)$.

We now set for each data structure $\mathcal{D}_i$ the number of updates until we recompute the data structure to $\Delta_i = n^{2/3} / 2^i (\log n)$. Observe that the total time to rebuild the data structures amortized over the number of updates can now be bound by 
\[
    \sum_{i = 0}^{\lceil \lg h \rceil} \frac{O(|C_i|n^2 h)}{\Delta_i} = O(n^{2+1/3} (\log n)^4 h)
\]
We can again deamortize by building the data structures in the background using standard techniques incurring at most an additional constant factor in the running time. A subtle detail is that if a data structure at level $i$ is replaced then all data structures at higher levels have to be replaced at the same time so that the sets $C_i$ form the hierarchy proposed in the preprocessing. Since rebuilding the data structures at higher levels can be done more efficiently, this however does only increase the running time by factor $2$.

Finally, let us discuss the delete procedure. Data structure $\mathcal{D}_1$ returns the shortest paths in $G \setminus D$ that does not contain a vertex in $C_1$ by Lemma \ref{lma:updateRandStructure}. Subsequently, each data structure $\mathcal{D}_i$ can be used to find the shortest paths through $C_i$ that do not contain a vertex in $C_{i+1}$. Combining these shortest paths by choosing the one of minimum weight for each tuple $(s,t) \in V^2$, we obtain all shortest paths that do not contain a vertex in $C_{\lceil \lg h \rceil + 1}$. Using Johnson's algorithm as described in Lemma \ref{lma:floydWarshall}, we can reinsert these vertices in time $O(|C_{\lceil \lg h \rceil + 1}|n^2)$ and handle the at most $\Delta_0$ insertions since the data structure $\mathcal{D}_0$ was last build. We can then return all-pairs shortest paths.

The update time to process the batch deletion in all data structures and the time spent to reinsert vertices can be bound by
\[
     (|C_{\lceil \lg h \rceil + 1}| + \Delta_0) * n^2 + \sum_{i = 0}^{\lceil \lg h \rceil} (\Delta_i \tau_i + n^3/h) = O(n^{2+2/3} (\log n)^3 + n^3 \log n/h)
\]
Setting $h = n^{1/3} (\log n)^2$ optimizes the running time. Finally, we point out that after the preprocessing step, each data structure is entirely deterministic. Since the initialization step also fixes a graph version once it starts that it keeps working on, the adversary cannot change the graph to affect the running time of the preprocessing step from that point. This implies that the algorithm even works against a non-oblivious adversary, that is an adversary that has access to the random bits used by the algorithms. We also point out that since we use the adapted Bellman-Ford procedure to store paths, the space of the data structure only differs by logarithmic factors from the space usage of the data structure presented in section \ref{subsec:spaceEfficientDet}.

\begin{corollary}
Let $G$ be an $n$-vertex directed edge-weighted graph undergoing vertex insertions and deletions. Then, there exists a Las-Vegas data structure which can maintain distances in $G$ between all pairs of vertices with update time $O(n^{2+2/3} (\log n)^3)$ w.h.p. using space $\tilde{O}(n^2)$ against a non-oblivious adversary. If the graph is unweighted, the running time can be improved to $O(n^{2+1/2} (\log n)^{3})$.
\end{corollary}

\section{Conclusion}

In this article, we present the first deterministic data structure that improves upon the longstanding result by Thorup \cite{thorup2005worst}. However, it remains open whether worst-case update time $\tilde{O}(n^{2+2/3})$ can also be achieved deterministically. We point out that one path to derandomize our last data structure is to obtain an \emph{amortized} update-time data structure that maintains hop-restricted shortest paths. Further a fundamental open problem is whether the worst-case update time can be further improved (or if lower bounds can rule out such an improvement). 

Finally, we provided the first space-efficient data structures for the dynamic APSP problem, i.e. the first data structures obtaining $\tilde{O}(n^2)$ space. Further progress towards an algorithm with $\tilde{O}(n^2)$ (amortized) update time and $\tilde{O}(n^2)$ space remains an important open problem.

\section*{Acknowledgements} The authors would like to thank Adam Karczmarz for pointing out an error in \cref{lemma:PebbleSumRand} and for providing a simple fix.

\printbibliography[heading=bibintoc] 

\pagebreak

\appendix
\section{Proof of lemma \ref{lemma:ExtendDists}}
\label{sec:proofLemmaExtDist}

We often use a simple lemma to compute hitting sets deterministically. 

\begin{lemma}[see \cite{thorup2005approximate, roditty2005deterministic}] \label{lma:Separator}
Let $N_1, N_2, \dots, N_n \subseteq U$ be a collection of subsets of $U$, with $u = |U|$ and $|N_i| \geq s$ for all $i \in [1,n]$. Then, we can implement a procedure $\textsc{Separator}(\{ N_i\}_{i \in [1,n]})$ that returns a set $A$ of size at most $O(\frac{u \log n}{s})$ with $N_i \cap A \neq \emptyset$ for all $i$, deterministically in $O(u + \sum_{i} |N_i|)$ time.
\end{lemma}

Below we describe the algorithm presented in \cite{Zwick02}. It is straight-forward to return paths, however, for simplicity, our algorithm only returns distances. Here, $i_{max} = \lceil \log_{3/2} n \rceil$.

\begin{algorithm}
\caption{$\textsc{DeterministicExtendDistances}(\Pi = \{ \pi_{i_h}(s,t) \}_{s,t}, h)$}
\label{alg:detExtendDists}
\KwIn{A collection of paths $\Pi$, that contains a path for each tuple $(s,t) \in V \times V$.}
\KwOut{Returns the set of distances $\{(\textsc{Dist}_{i_{max}}(s,t)\}_{s,t \in V \times V}$.}
\BlankLine
\ForEach{$(s,t) \in V \times V$}{
    $\textsc{Dist}_{i_h}(s,t) \gets w(\pi_{i_h}(s,t))$
}

\For(\label{lne:DetExtendDistsi}){$i \gets i_h + 1$ to $i_{max}$}{
    Compute a set $\textsc{Separator}$ of size $O(n \log n / h_i)$ that contains a vertex from each path in $\Pi$ of hop at least $\lfloor \frac{1}{4} h_i\rfloor$.\label{lne:DetExtendDistsSeparator}\;
    \ForEach{$(s,t) \in V \times V$}{
        $\textsc{Dist}_{i}(s,t) \gets \textsc{Dist}_{i-1}(s,t)$\label{lne:DetExtendDistsSingleAssignment}\;
        \ForEach{$x \in \textsc{Separator}$}{
            $\textsc{Dist}_{i}(s,t) \gets \min\{\textsc{Dist}_{i}(s,t), \textsc{Dist}_{i-1}(s,x) + \textsc{Dist}_{i-1}(x,t)\}$\label{lne:DetExtendDistsSumAssignment}
        }
    }
}
\Return $\{(\textsc{Dist}_{i_{max}}(s,t)\}_{s,t \in V \times V}$
\end{algorithm}

\begin{lemma}[Restatement of \ref{lemma:ExtendDists}]
Given a collection $\Pi$ of the $h$-hop-improving shortest paths for all pairs $(s,t) \in V^2$ in $G \setminus D$, then there exists a procedure \ExtendDists that returns improving shortest path distances for all pairs $(s,t) \in V^2$ in time $O(n^3 \log n/ h + n^2 \log^2 n)$. 
\end{lemma}
\begin{proof}
Let us denote by $\textsc{Separator}^i$ the separator at phase $i$. We compute the initial separator $\textsc{Separator}^{i_h + 1}$ in time $O(n^2 h)$ time by finding a hitting set for all paths of hop at least $\frac 1 4 h$. For  $\textsc{Separator}^i$ with $i > i_h + 1$, we only hit the shortest paths from and to all vertices in $\textsc{Separator}^{i-1}$ of length $\frac 1 4 h_i$. This suffices since the $\textsc{Separator}^{i-1}$ "hit" all paths of length $\frac 1 4 h_{i-1}$ so all paths of hop $\frac 1 4 h_i$ have to go through at least one such vertex. The running time to compute the separator at phase $i$ is thus $O(n \log n/h_i * n h_i) = O(n^2 \log n)$. Using this layering the initial separator $\textsc{Separator}^{i_h + 1}$ can also be computed in time $O(n^2 \log n)$. Once, the separator is computed, each path at iteration $i$ can be computed in time $O(|\textsc{Separator}^i|)$. Using a geometric sum argument, we obtain the claimed running time.

The correctness of the lemma follows if we can establish the loop-invariant for the for-loop in line \ref{lne:DetExtendDistsi} that at the beginning of iteration $i$, for each pair $(s,t)$, if there is a shortest path from $s$ to $t$ in $G \setminus D$ which has hop at most $h_{i-1}$ then $\textsc{Dist}_{i-1}(s,t)$ is the weight of a shortest path from $s$ to $t$ in $G\setminus D$. Here, we prove the claim only for shortest paths, however, an extension to improving shortest-paths is straight-forward.

The proof is by induction on $i\geq i_h+1$. When $i = i_h+1$, the claim is true by our assumption on $\Pi$. Now consider the beginning of an iteration $i\leq i_{\max}$ and assume that the invariant holds at this point. Let $s,t\in V$ be given and assume that there is a shortest path $P$ from $s$ to $t$ in $G\setminus D$ such that $P$ has hop at most $h_i$. Pick $P$ such that its hop is minimized.

If $|P|\leq \frac 2 3 h_i = h_{i-1}$ then in line~\ref{lne:DetExtendDistsSingleAssignment}, $\textsc{Dist}_i(s,t) = \textsc{Dist}_{i-1}(s,t) = w(P) = d_{G\setminus D}(s,t)$ by the induction hypothesis. Since $\textsc{Dist}_i(s,t)$ can never increase, the invariant holds at the beginning of iteration $i+1$ for $s$ and $t$.

Now, assume that $|P| > \frac 2 3 h_i$. Since $i\geq i_h+1$, we have $|P| > \frac 2 3 h_{i_h+1} = h = h$. Let $a$ be the vertex of $P$ such that $|P[s,a]| = \lceil\frac 1 2 |P| - \frac 1 8 h\rceil$ and let $b$ be the vertex of $P[a,t]$ such that $|P[a,b]| = \lfloor\frac 1 4 h\rfloor$. Since $P[a,b]$ is a shortest path of minimum hop, the induction hypothesis implies that $P' = P[s,a]\circ\pi_{i_h}(a,b)\circ P[b,t]$ is a shortest path from $s$ to $t$ in $G\setminus D$ of minimum hop $|P'| = |P|$. Since $|\pi_{i_h}(a,b)| = |P[a,b]| = \lfloor\frac 1 4 h\rfloor$, $\pi_{i_h}(a,b)$ intersects $\textsc{Separator}$ in a vertex $x$. Since $|P'[s,x]|\leq \lceil\frac 1 2 |P| - \frac 1 8 h\rceil + \lfloor\frac 1 4 h\rfloor\leq \frac 2 3 |P|\leq \frac 2 3 h_i = h_{i-1}$, we have $\textsc{Dist}_{i-1}(s,x) = w(P'[s,x])$. Similarly, since $|P'[x,t]|\leq |P| - \lceil\frac 1 2 |P| - \frac 1 8 h\rceil\leq\frac 2 3 |P|\leq h_{i-1}$, we have $\textsc{Dist}_{i-1}(x,t) = w(P'[x,t])$. Hence, in line~\ref{lne:DetExtendDistsSumAssignment}, $\textsc{Dist}_i(s,t)$ is set to $w(P'[s,x]) + w(P'[x,t]) = w(P') = \textsc{Dist}_{G\setminus D}(s,t)$. This shows that the invariant holds at the beginning of iteration $i+1$ for $s$ and $t$, as desired.

At termination, $i = i_{\max}+1$ and the invariant states that for all $s,t\in V$, if there is a shortest path from $s$ to $t$ in $G\setminus D$ and this path has at most $h_{i_{\max}}\geq n$ hops then $\textsc{Dist}_{i_{\max}}(s,t) = \textsc{Dist}_{G\setminus D}(s,t)$. The lemma now follows since any simple shortest path has at most $n-1$ hops.
\end{proof}

\end{document}